 \documentclass{llncs}
\usepackage[symbol]{footmisc}
\usepackage{mathtools}
\usepackage{amsmath}
\usepackage{amssymb}
\usepackage{hyperref}
\usepackage[table, dvipsnames]{xcolor}
\usepackage{adjustbox}
\usepackage{array, makecell}
\usepackage{framed}
\usepackage{graphicx}
\usepackage{bbm}
\usepackage{amsfonts}
\usepackage[nameinlink]{cleveref}
\Crefname{algocf}{Algorithm}{Algorithms}
\crefname{algocfline}{line}{lines}

\usepackage[ruled,vlined,linesnumbered,algonl]{algorithm2e}
\SetEndCharOfAlgoLine{}
\usepackage{breakcites}
\usepackage{scalerel}[2016/12/29] % To scale z in exponent
\usepackage{multirow}
\usepackage{tikz}

\SetCommentSty{mycommfont}
\DontPrintSemicolon
\usepackage{newfloat}
\DeclareFloatingEnvironment[fileext=lop]{Algorithm}

\newcommand{\X}{\mathcal{X}}
\newcommand{\veps}{\varepsilon}
\newtheorem{defn}[theorem]{Definition}
\newtheorem{clm}{Claim}

\Crefname{clm}{Claim}{Claims}
\crefname{clm}{Claim}{Claims}

\newcommand{\C}{\mathcal{C}}
\newcommand{\A}{\mathcal{A}}

\newcommand{\M}{\mathcal{M}}
\newcommand{\cI}{{\mathcal I}}

\newcommand{\pr}{\mathbf{Pr}}

\newcommand{\kmed}{{\textsf{$k$-median}}\xspace}
\newcommand{\kmean}{{\textsf{$k$-means}}\xspace}
\newcommand{\rgath}{{\textsf{$r$-gathering}}\xspace}
\newcommand{\chck}{{\textsf{check}}}
\newcommand{\opt}{{\textsf{opt}}}
\newcommand{\cost}{{\textsf{cost}}}
\newcommand{\lmatch}{{\textsf{$b$-matching} }}
\newcommand{\btau}{{\boldsymbol{\tau}}}
\newcommand{\Xopt}{X^{\opt}}
\newcommand{\fopt}{f^{\opt}}

\newcommand{\hatX}{{\widehat X}}
\newcommand{\hatt}{{\hat t}}
\newcommand{\cO}{{\cal O}}

\newcommand{\hmu}{{\widehat{\mu}}}

\definecolor{mycolor}{rgb}{1, 0.0, 0.0}
\newcommand{\dishant}[1]{\textcolor{mycolor}{}}

\allowdisplaybreaks

\begin{document}

\title{Clustering What Matters in Constrained Settings}
%\lv{
\subtitle{(Improved Outlier to Outlier-Free Reductions)}
%}
%
%\titlerunning{}
% abbreviated title (for running head)
%                                     also used for the TOC unless
%                                     \toctitle is used
%
\author{Ragesh Jaiswal and Amit Kumar}
%
%\authorrunning{Jaiswal et al.}
% abbreviated author list (for running head)
%
%%%% list of authors for the TOC (use if author list has to be modified)
%\tocauthor{Ragesh Jaiswal}
%
\institute{Department of Computer Science and Engineering, Indian Institute of Technology Delhi.\thanks{\email{\{rjaiswal, amitk\}@cse.iitd.ac.in}}}
{\def\addcontentsline#1#2#3{}\maketitle}
%\maketitle     % typeset the title of the contribution
%%\thispagestyle{empty}

%TODO mandatory: add short abstract of the document
\begin{abstract}
Constrained clustering problems generalize classical clustering formulations, e.g., \kmed, \kmean, by imposing additional constraints on the feasibility of a clustering. There has been significant recent progress in obtaining approximation algorithms for these problems, both in the metric and the Euclidean settings. However, the outlier version of these problems, where the solution is allowed to leave out $m$ points from the clustering, is not well understood. In this work, we give a general framework for reducing the outlier version of a constrained \kmed or \kmean problem to the corresponding outlier-free version with only $(1+\varepsilon)$-loss in the approximation ratio. The reduction is obtained by mapping the original instance of the problem to $f(k,m, \varepsilon)$ instances of the outlier-free version, where $f(k, m, \veps) = \left( \frac{k+m}{\veps}\right)^{O(m)}$. 
As specific applications, we get the following results:
\begin{itemize}
    \item First FPT ({\it in the parameters $k$ and $m$}) $(1+\varepsilon)$-approximation algorithm for the outlier version of capacitated \kmed and \kmean in Euclidean spaces with {\it hard} capacities. 
    \item First FPT ({\it in the parameters $k$ and $m$}) $(3+\veps)$ and $(9+\veps)$ approximation algorithms for the outlier version of capacitated \kmed and \kmean, respectively,  in general metric spaces with {\it hard} capacities.
    \item First FPT ({\it in the parameters $k$ and $m$}) $(2-\delta)$-approximation algorithm for the outlier version of the \kmed problem under the Ulam metric.
\end{itemize}

Our work generalizes the results of Bhattacharya et al. 
%~\cite{bgjk20} 
and  Agrawal et al.
%~\cite{aisx23} 
to a larger class of constrained clustering problems. Further, our reduction works for arbitrary metric spaces and so can extend clustering algorithms for outlier-free versions in both Euclidean and arbitrary metric spaces. 
\end{abstract}

\section{Introduction}
Center-based clustering problems such as $k$-median and the $k$-means are important data processing tasks. Given a metric $D$ on a set of $n$ points $\X$ and a parameter $k$, the goal here is to partition the set of points into $k$ {\it clusters}, say $C_1, \ldots, C_k$,  and assign the points in each cluster to a corresponding {\it cluster center}, say $c_1, \ldots, c_k$, respectively,  such that the objective $\sum_{i=1}^k \sum_{x \in C_i} D(x,c_i)^z$ is minimized. Here $z$ is a parameter which is 1 for $\kmed$ and 2 for $\kmean$. The {\it outlier} version of these problems is specified by another parameter $m$, where a solution is allowed to leave out up to $m$ points from the clusters. Outlier versions capture settings where the input may contain a few highly erroneous data points. Both the outlier and the outlier-free versions have been well-studied in the literature with constant factor approximations known for both the $\kmean$ and the $\kmed$ problem~\cite{Svensson17,naveen02,charikar99}. In addition, fixed-parameter tractable (FPT)  $(1+\varepsilon)$-approximation algorithms are known for these problems in the Euclidean setting~\cite{kumar02,feldman07,bgjk20}: the running time of such algorithms is of the form $f(k,m,\varepsilon) \cdot poly(n,d)$, where $f()$ is an exponential function of the parameters $k,m,\varepsilon$ and $d$ denotes the dimensionality of the points.

A more recent development in clustering problems has been the notion of {\it constrained clustering}. A constrained clustering problem specifies additional conditions on a feasible partitioning of the input points into $k$ clusters. For example, the $\rgath$ problem requires that each cluster in a feasible partitioning must contain at least $r$ data points. Similarly, the well-known {\it capacitated} clustering problem specifies an upper bound on the size of each cluster.  Constrained clustering formulations can also capture various types of {\it fairness} constraints: each data point has a {\it label} assigned to it, and we may require upper or lower bounds on the number (or fraction) of points with a certain label in each cluster.~\Cref{table:1} in the Appendix gives a list of some of these problems.  FPT (in the parameter $k$) constant factor approximation algorithms are known for a large class of these problems (see~\Cref{table:2} in the Appendix).

It is worth noting that constrained clustering problems are distinct from outlier clustering: the former restricts the set of feasible partitioning of input points, whereas the latter allows us to reduce the set of points that need to be partitioned into clusters. 
There has not been much progress on constrained clustering problems in the outlier setting (also see~\cite{outlier:kmeans_2018_Ravishankar} for 
unbounded integrality gap for the natural LP relaxation for the outlier clustering versions).
In this work, we bridge this gap between the outlier and the outlier-free versions of constrained clustering problems by giving an {\it almost approximation-preserving} reduction from
the former to the latter. As long as the parameters of interest (i.e., $k,m$) are small, the reduction works in polynomial time. Using our reduction, an FPT $\alpha$-approximation algorithm for the outlier-free version of a constrained clustering problem leads to an FPT $(\alpha+\varepsilon)$-approximation algorithm for the outlier version of the same problem. For general metric spaces, this implies the first FPT constant-approximation for outlier versions of several constrained clustering problems; and similarly, we get new FPT $(1+\varepsilon)$-approximation algorithms for several outlier constrained clustering problems --see~\Cref{table:2} in the Appendix for the precise details. 

This kind of FPT approximation preserving reduction in the context of Euclidean $k$-means was first given by~\cite{bgjk20} using a sampling-based approach. 
~\cite{gjk20} extended the sampling ideas of~\cite{bgjk20} to general metric spaces but did not give an approximation-preserving reduction.
~\cite{aisx23} gave a reduction for general metric spaces using a coreset construction. 
In this work, we use the sampling-based ideas of \cite{bgjk20} to obtain an approximation-preserving reduction from the outlier version to the outlier-free version with improved parameters over~\cite{aisx23}. Moreover, our reduction works for most known constrained clustering settings as well. 

\subsection{Preliminaries}
We give a general definition of a constrained clustering problem.  For a positive integer $k$, we shall use $[k]$ to denote the set $\{1, \ldots, k\}$. 
Let $(\X, D)$ denote the metric space with distance function $D$. For a point $x$ and a subset $S$ of points, we shall use $D(x,S)$ to denote $\min_{y \in S} D(x,y)$. 
The set $\X$  contains subsets $F$ and $X$: here $X$ denotes the set of input points and $F$ is the set of points where a center can be located. An outlier constrained clustering problem is specified by the following parameters and functions:
\begin{itemize}
    \item $k$: the number of clusters. 
    \item $m$: the number of points which can be left out from the clusters. 
    \item a function $\chck$: given a partitioning $X_0, X_1, \ldots, X_k$ of $X$ (here $X_0$ is the set of outliers) and centers $f_1, \ldots, f_k$, each lying in the set $F$, the function $\chck(X_0, X_1, \ldots, X_k, f_1, \ldots, f_k)$ outputs $1$ iff this is a feasible clustering. For example, in the $\rgath$ problem, the $\chck(X_0, X_1, \ldots, X_k, f_1, \ldots, f_k)$ outputs $1$ iff $|X_i| \geq r$ for each $i \in [k]$.  The $\chck$ function depends only on the cardinality of the sets $X_1, \ldots, X_k$ and the locations $f_1, \ldots, f_k$. This already captures many of the constrained clustering problems. Our framework also applies to the more general labelled version (see details below). 
    \item a cost function $\cost$: given a partitioning $X_0, X_1, \ldots, X_k$ of $X$ and centers $f_1, \ldots, f_k$, $$\cost(X_0, X_1, \ldots, X_k, f_1, \ldots, f_k) := \sum_{i \in [k]} \sum_{x \in X_i} D^z(x,f_i), $$
    where $z$ is either 1 (the outlier constrained $\kmed$ problem) or 2 (the outlier constrained $\kmean$ problem). 
\end{itemize} 
Given an instance $\cI = (X, F, k, m, \chck, \cost)$ of an outlier constrained clustering problem as above, the goal is to find a partitioning $X_0, X_1, \ldots, X_k$ of $X$ and centers $f_1, \ldots, f_k \in F$ such that $|X_0| \leq m$,  \\
$\chck(X_0, X_1,\ldots, X_k, f_1, \ldots, f_k)$ is 1 and 
$\cost(X_0, X_1, \ldots, X_k, f_1, \ldots, f_k)$ is minimized. The outlier-free constrained clustering problem is specified as above, except that the parameter $m$ is 0. For the sake of brevity, we leave out the parameter $m$ and the set $X_0$ while defining the instance $\cI$, and functions $\chck$ and $\cost$.

We shall also consider a more general class of constrained clustering problems, where each input point is assigned a {\it label}. In other words, an instance $\cI$ of such a problem is specified by a tuple $(X,F,k,m, \sigma, \chck, \cost)$, where $\sigma: X \rightarrow L$ for a finite set $L$. Note that the $\chck$ function may depend on the function $\sigma$. For example, $\sigma$ could assign a label ``red'' or ``blue'' to each point in $X$ and the $\chck$ function would require that each cluster $X_i$ should have an equal number of red and blue points. In addition to the locations $f_1, \ldots, f_k$, the $\chck(X_1, \ldots, X_k, f_1, \ldots, f_k, \sigma)$ function also depends on $|\sigma^{-1}(l) \cap X_j|$ for each $l \in L, j \in [k]$, i.e., the number of points with a particular label in each of the clusters. Indirectly, this also implies that the $\chck$ function can impose conditions on the labels of the outliers points. For example, 
 the colorful $k$-median problem discussed in \cite{aisx23} has the constraint that $m_i$ clients from the label type $i$ should be designated as outliers, given that every client has a unique label. %Note that our technique will apply to this setting since we replace outliers with the same labelled client ({\it after guessing the labels of the outliers}). 
%Table~\ref{table:2} summarises the results for the problems in Table~\ref{table:1}.
~\Cref{table:1} in the Appendix gives a description of some of these problems.

We shall use the approximate triangle inequality, which states that for $z \in \{1,2\}$ and any three points $x_1, x_2, x_3 \in \X$, 
\begin{align}
    \label{eq:tr}
    D^z(x_1, x_3) \leq z \left( D^z(x_1, x_2) + D^z(x_2, x_3) \right).
\end{align}

\subsection{Our results}
Our main result reduces the outlier constrained clustering problem to the outlier-free version. In our reduction, we shall also use approximation algorithms for the (unconstrained) $\kmed$ and $\kmean$ problems. We assume we have a constant factor approximation algorithm for these problems\footnote{Several such constant factor approximation algorithms exist~\cite{Svensson17,naveen02,charikar99}.}: let $\C$ denote such an algorithm with running time $T_{\C}(n)$ on an input of size $n$. Note that $\C$ would be an algorithm for the $\kmean$ or the $\kmed$ problem depending on whether $z=1$ or $2$ in the definition of the $\cost$ function. 

\begin{theorem}[Main Theorem]\label{thm:main}
Consider an instance $\cI=(X,F,k,m,\chck, \cost)$ of an outlier constrained clustering problem. Let $\A$ be an $\alpha$-approximation algorithm for the corresponding outlier-free constrained clustering problem; let $T_{\A}(n)$ be the running time of $\A$ on an input of size $n$.  Given a positive $\varepsilon > 0$, there is an $\alpha(1+\varepsilon)$-approximation algorithm for $\cI$ with running time 
%$T_{\C}(n)  + q \cdot T_{\A}(n) + O\left( \frac{qn(k+m)m \log{m}}{\veps^z}\right)$, 
$T_{\C}(n)  + q \cdot T_{\A}(n) + O \left( n \cdot (k + \frac{m^{z+1} \log{m}}{\veps^z})\right) + O \left( qm^2(k+m)^3\right)$,
where $n$ is the size of $\cI$ and $q = f(k, m, \veps) = \left( \frac{k+m}{\veps}\right)^{O(m)}$, and $z=1$ or $2$ depending on the $\cost$ function (i.e., $z=1$ for $\kmed$ objection and $z=2$ for $\kmean$ objective).
\end{theorem}

The above theorem implies that as long as there is an FPT or polynomial-time approximation algorithm for the constrained, outlier-free $\kmed$ or $\kmean$ clustering problem, there is an FPT approximation algorithm (with almost the same approximation ratio) for the corresponding outlier version. We prove this result by creating $q$ instances of the outlier-free version of $\cI$ and picking the best solution on these instances using the algorithm $\A$. 
We also extend the above result to the labelled version: 

\begin{theorem}[Main Theorem: labelled version]\label{thm:main-labelled}
Consider an instance $\cI=(X,F,k,m, \sigma, \chck, \cost)$ of an outlier constrained clustering problem with labels on input points. Let $\A$ be an $\alpha$-approximation algorithm for the corresponding outlier-free constrained clustering problem; let $T_{\A}(n)$ be the running time of $\A$ on an input of size $n$.  Given a positive $\varepsilon > 0$, there is an $\alpha(1+\varepsilon)$-approximation algorithm for $\cI$ with running time %$T_{\C}(n)  + q \cdot T_{\A}(n) + O\left( \frac{qn(k+m)m \log{m}}{\veps^z}\right)$, 
$T_{\C}(n)  + q \cdot T_{\A}(n) + O \left( n \cdot (k + \frac{m^{z+1} \log{m}}{\veps^z})\right) + O \left( q \ell m^2(k+m)^3\right)$,
where $n$ is the size of $\cI$, $q = f(k, m, \veps) = \left( \frac{(k+m)\ell}{\veps}\right)^{O(m)}$ with $\ell$ being the number of distinct labels, and $z=1$ or $2$ depending on the $\cost$ function (i.e., $z=1$ for $\kmed$ objection and $z=2$ for $\kmean$ objective). 
\end{theorem}

The algorithms given in Theorem~\ref{thm:main} and Theorem~\ref{thm:main-labelled} are randomized algorithms that guarantee the stated approximation factor with high probability.
The consequences of our results for specific constrained clustering problems are summarized in~\Cref{table:2} in the Appendix. We give the results of related works~\cite{bgjk20,gjk20,aisx23} in the same table to see the contributions of this work. Our contributions can be divided into two main categories:
\begin{enumerate}
\item {\it Matching the best-known result}: This can be further divided into two categories: 
\begin{enumerate}
\item {\it Matching results of \cite{aisx23}}: \cite{aisx23} gives an outlier to outlier-free reduction. We also give such a reduction using a different technique with better parameters. This means that we match all the results of \cite{aisx23}, which includes problems such as the classical $k$-median/means problems, the Matroid $k$-median problem, the colorful $k$-median problem, and $k$-median in certain special metrics. See rows 2-6 in~\Cref{table:2} given in the Appendix.
\item  {\it Matching results of \cite{gjk20}}: \cite{gjk20} gives FPT approximation algorithms for certain constrained problems on which the coreset-based approach of \cite{aisx23} is not known to work. See the last row of~\Cref{table:2}. \cite{gjk20} gives algorithms for outlier and outlier-free versions with the same approximation guarantee. 
Since the best outlier-free approximation is also from \cite{gjk20}, our results currently only match the approximation guarantees of \cite{gjk20}. However, if there is an improvement in any of these problems, our results will immediately beat the known outlier results of \cite{gjk20}.
\end{enumerate}
\item {\it Best known results}: Since our results hold for a larger class of constrained problems than earlier works, there are certain problems for which our results give the best-known FPT approximation algorithm. The list includes capacitated \kmed/\kmean with hard capacities in general metric and Euclidean spaces. It also includes the $k$-median problem in the Ulam metric. A recent development in the Ulam $\kmed$ problem~\cite{cdk23} has broken the $2$-approximation barrier. Our reduction allows us to take this development to the outlier setting as well. 
The outlier-free results from which our best results are derived using our reduction are given in~\Cref{table:2} (see rows 7-9) given in the Appendix.
\end{enumerate}

\subsection{Comparison with earlier work}
As discussed earlier, the idea of a reduction from an outlier clustering problem to the corresponding outlier-free version in the context of the Euclidean \kmean problem was suggested by~\cite{bgjk20} using a $D^2$-sampling based idea. 
\cite{gjk20} used the sampling ideas to design approximation algorithms for the outlier versions of various constrained clustering problems. However, the approximation guarantee obtained by~\cite{gjk20} was limited to $(3+\veps)$ for a large class of constrained \kmed and $(9+\veps)$ for the constrained \kmean problems, and it was not clear how to extend these techniques to get improved guarantees. 
As a result, their techniques could not 
 exploit the recent developments by \cite{vincent_hardness_2019} in the design of $(1+2/e+\veps)$ and $(1+8/e+\veps)$ FPT approximation algorithms for the classical outlier-free \kmed and \kmean problems respectively in general metric spaces.
\cite{aisx23} gave an outlier-to-outlier-free reduction, making it possible to extend the above-mentioned FPT approximation guarantees for the outlier-free setting to the outlier setting.

The reduction of \cite{aisx23} is based on the coreset construction by \cite{chen09} using uniform sampling. A coreset for a dataset is a weighted set of points such that the clustering of the coreset points with respect to any set of $k$ centers is the same (within a $1\pm \veps$ factor) as that of the original set points. 
The coreset construction in~\cite{chen09} starts with a set $C$ of centers that give constant factor approximation. 
They consider $O(\log{n})$ ``ring'' around these centers, uniformly sample points from each of these rings, and set the weight of the sampled points appropriately. The number of sampled points, and hence the size of the coreset, is $\left( \frac{|C| \log{n}}{\veps} \right)^2$. 
\cite{aisx23} showed that when starting with $(k+m)$ centers that give a constant approximation to the classical $(k+m)$-median problem, the coreset obtained as above has the following additional property:  for any set of $k$ centers, the clustering cost of the original set of points excluding $m$ outliers is  same (again, within $1\pm \veps$ factor) as that of the coreset, again allowing for exclusion of a subset of $m$ points from it. 
This means that by trying out all $m$ subsets from the coreset, we ensure that at least one subset acts as a good outlier set. Since the coreset size is $\left( \frac{(k+m)\log{n}}{\veps}\right)^2$, the number of outlier-free instances that we construct is $\left( \frac{(k+m)\log{n}}{\veps}\right)^{O(m)}$. Using $(\log{n})^{O(m)} = \max\{m^{O(m)}, n^{O(1)}\}$, this is of the form $f(k, m, \veps) \cdot n^{O(1)}$ for a suitable function $f$.  
At this point, we note the first quantitative difference from our result. 
In our algorithm, we save the $(\log{n})^{O(m)}$ factor, which also means that the number of instances does not depend on the problem size $n$. 
Further, a coreset-based construction restricts the kind of problems it can be applied to.
The coreset property that the cost of original points is the same as that of the  weighted cost of coreset points holds when points are assigned to the closest center ({\it i.e., the entire weight of the coreset goes to the closest center}).\footnote{The reason is how Haussler's lemma is applied to bound the cost difference.} 
This works for the classical unconstrained \kmed and \kmean problems (as well as the few other settings considered in \cite{aisx23}).
However, for several constrained clustering problems, it may not hold that every point is assigned to the closest center. 
There have been some recent developments~\cite{bfs21,braverman22} in designing coresets for constrained clustering settings. 
However, they have not been shown to apply to the outlier setting.
Another recent work~\cite{huang22} designs coresets for the outlier setting, but like \cite{aisx23}, it has limited scope and has not been shown to extend for most constrained settings.
Our $D^z$-sampling-based technique has the advantage that instead of running the outlier-free algorithm on a coreset as in \cite{aisx23}, it works directly with the dataset. 
That is, we run the outlier-free algorithm on the dataset (after removing outlier candidates). 
This also makes our results helpful in weighted settings (e.g., see \cite{cdk23}) where the outlier-free algorithm is known to work only for unweighted datasets -- note a coreset is a weighted set).

\noindent
{\bf Recent independent work}: In recent and independent work,~\cite{dabas2023fpt} design similar approximation preserving reductions for a restricted class of constrained clustering settings, namely capacitated clustering and $(\alpha, \beta)$-fair clustering. Further, their results are obtained by extending coreset based ideas of ~\cite{aisx23}.

\subsection{Our Techniques}
In this section, we give a high-level description of our algorithm. Let $\cI$ denote an instance of outlier constrained clustering on a set of points $X$ and $\cO$ denote an optimal solution to $\cI$. The first observation is that the optimal cost of the outlier-free and unconstrained clustering with $k+m$ centers on $X$ is a lower bound on the cost of $\cO$ (\Cref{cl:1}).
\footnote{This observation was used in both \cite{bgjk20} and \cite{aisx23}.} 
Let $C$ denote the set of these $(k+m)$ centers (we can use any constant factor approximation for the unconstrained version to find $C$). The intuition behind choosing $C$ is that the centers in $\cO$ should be close to $C$. 

Now we divide the set of $m$ outliers in $\cO$ into two subsets: those which are far from $C$ and the remaining ones close to $C$ (``near'' outliers). Our first idea is to randomly sample a subset $S$ of $O(m \log m)$ points from $X$ with sampling probability proportional to distance (or square of distance) from the set $C$. This sampling ensures that $S$ contains the far outliers with high probability (\Cref{cl:F}). We can then iterate over all subsets of $S$ to guess the exact subset of far outliers. Handling the near outliers is more challenging and forms the heart of the technical contribution of this paper. 

We ``assign'' each near outlier to its closest point in $C$ -- let $\Xopt_{N,j}$ be the set of outliers assigned to $c_j$. By iterating over all choices, we can guess the cardinality $t_j$ of each of the sets $\Xopt_{N,j}$. We now set up a suitable minimum cost bipartite $b$-matching instance which assigns a set of $t_j$ points to each center $c_j$. Let $\hatX_j$ be the set of points assigned to $c_j$. Our algorithm uses $\cup_j \hatX_j$ as the set of near outliers. In the analysis, we need to argue that there is a way of matching the points in $\Xopt_{N,j}$ to $\hatX_j$ whose total cost (sum of distances or squared distances between matched points) is small (\Cref{lem:mu}). The hope is that we can go from the optimal set of outliers in $\cO$ to the ones in the algorithm and argue that the increase in cost is small. Since we are dealing with constrained clustering, we need to ensure that this process does not change the size of each of the clusters. To achieve this, we need to further modify the matching between the two sets of outliers (\Cref{lem:hmu}). Finally, with this modified matching, we are able to argue that the cost of the solution produced by the algorithm is close to that of the optimal solution. The extension to the labelled version follows along similar lines.

In the remaining paper, we prove our two main results, Theorem~\ref{thm:main} and Theorem~\ref{thm:main-labelled}.
The main discussion will be for Theorem~\ref{thm:main} since Theorem~\ref{thm:main-labelled} is an extension of Theorem~\ref{thm:main} that uses the same proof ideas.
In the following sections, we give the details of our algorithm (Section \ref{sec:algorithm}) and its analysis (Section~\ref{sec:analysis}). In Section~\ref{sec:labelled}, we discuss the extension to the labelled version.

\section{Algorithm}\label{sec:algorithm}
In this section, we describe the algorithm for the outlier constrained clustering problem. Consider an instance  $\cI = (X,F,k,m,\chck, \cost)$ of this problem. Recall that the parameter $z=1$ or $2$ depends on whether the $\cost$ function is like the $\kmed$ or the $\kmean$ objective respectively. In addition, we assume the existence of the following algorithms:
\begin{itemize}
    \item A constant $\beta$-factor algorithm $\C$ for the $\kmed$ or the $\kmean$ problem (depending on $z=1$ or $z=2$ respectively): an instance here is specified by a tuple $(X', F',k')$ only, where $X'$ is the set of input points, $F'$ is the set of potential locations for a center, and $k'$ denotes the number of clusters. 
%We use $\C$ to denote this algorithm and $\beta$ its approximation guarantee.
    \item An algorithm $\A$ for the outlier-free version of this problem. An instance here is given by a tuple $(X',F', k, \chck, \cost)$ where the $\chck$ and the $\cost$ functions are the same as those in $\cI$. 
    \item An algorithm $\M$ for the $\lmatch$ problem: an instance of the $\lmatch$ problem is specified by a weighted bi-partite graph $G=(L,R = \{v_1, \ldots, v_r\} ,E)$, with edge $e$ having weight $w_e$; and a tuple $(t_1, \ldots, t_r)$, where $t_i, i \in [r], $ are non-negative integers.  A solution needs to find a subset of $E'$ of $E$ such each vertex of $L$ is incident with at most one edge of $E'$, and each vertex $v_j \in R$ is incident with {\it exactly} $t_j$ edges of $E'$. The goal is to find such a set $E'$ of minimum total weight. 
\end{itemize}

\noindent
We now define $D^z$-sampling:
\begin{defn}
Given sets $C$ and $X$ of points, {\it $D^z$-sampling from $X$ w.r.t. $C$} samples a point $x \in X$, where the probability of sampling $x$ is proportional to $D^z(x,C)$.
\end{defn}

\begin{figure}[t!]
\centering
\includegraphics[scale=0.30]{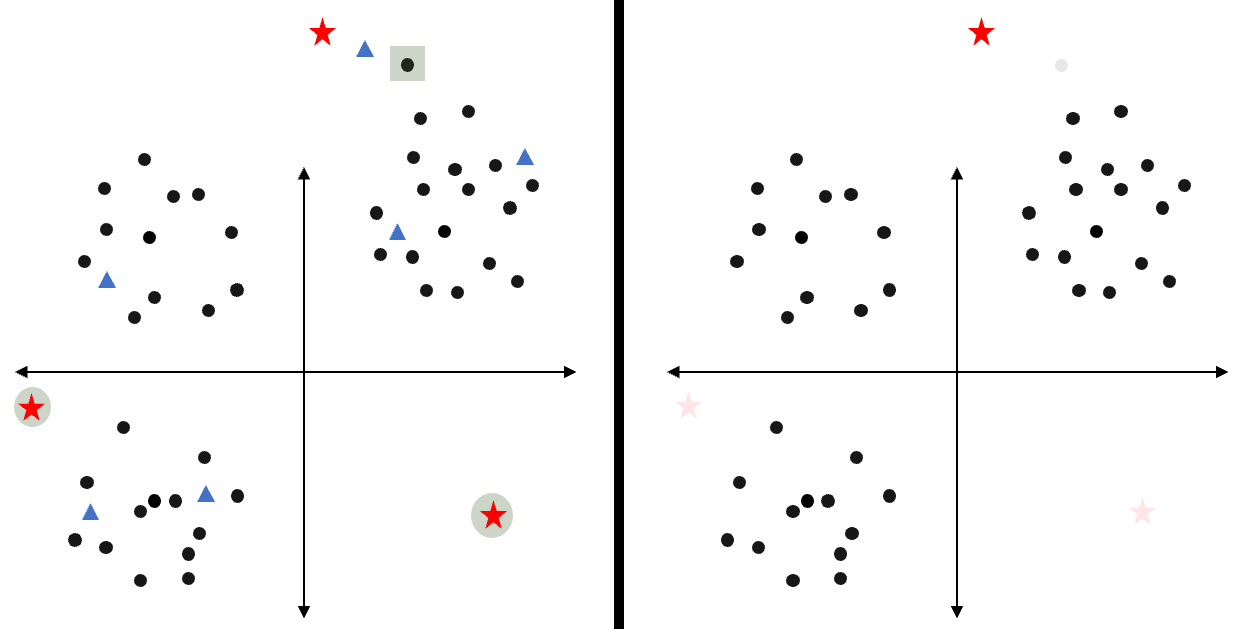}
\caption{An example 2-dimensional instance with ($k=3; m = 3; F = C$), where the red stars are the optimal outliers. The reduction algorithm finds a set $C$ of $k+m = 6$ centers ({\em shown as blue triangles}). It then $D^z$-samples $O(m \log{m})$ points with respect to center set $C$, which guarantees that the faraway outliers ({\em see red stars shaded with green circles}) are found. The outliers near $C$ ({\em see the top red star}) are not discovered this way. So, we find a suitable ``replacement or proxy'' ({\em see the point shaded with green square}) for such outliers by setting up a $b$-matching problem to locate suitable points that are close to the centers in $C$. The instance for the outlier-free version is obtained by removing a suitable subset of proxies and faraway outliers from the point set ({\em see figure on the right}). The key technicality lies in showing that designating proxies as outliers does not increase the cost too much. %Intuitively, this should hold since proxies are close to their true outliers.
}
\label{fig:algorithm-outline}
\end{figure}

See Figure~\ref{fig:algorithm-outline} for a high-level outline of the algorithm.
The algorithm is described in~\Cref{algo:cluster}. It first runs the algorithm $\C$ to obtain a set of $(k+m)$ centers $C$ in line~\ref{l:C}. In line~\ref{l:D2}, we sample a subset $S$ where  each point in $S$ is sampled independently using $D^z$-sampling w.r.t. $C$.  Given a subset $Y$, we say that a tuple $\btau=(t_1, \ldots, t_{k+m})$ is {\it valid w.r.t. $Y$}  if $t_j \geq 0$ for all $j \in [k+m]$, and $\sum_{j} t_j + |Y| = m$. For each subset $Y$ of size $\leq m$ of $S$ and for each valid tuple $\btau$, the algorithm constructs a solution $(X^{(Y,\btau)}_0, X^{(Y,\btau)}_1, \ldots, X^{(Y,\btau)}_k)$, where $X^{(Y,\btau)}_0$ denotes the set of outlier points. This is done by first computing the set $X^{(Y,\btau)}_0$, and then using the algorithm $\A$ on the remaining points $X \setminus (X_0^{(Y, \btau)} \cup Y)$ (line~\ref{l:rest}). To find the set $X^{(Y,\btau)}_0$, we construct an instance $\cI^{(Y, \btau)}$ of $\lmatch$ first (line~\ref{l:match}). This instance is defined as follows: the bipartite graph has the set of $(k+m)$ centers $C$ on the right side and the set of points $X$ on the left side. The weight of an edge between a vertex $v \in C$ and $w \in X$ is equal to $D^z(v,w)$. For each vertex $v_j \in C$, we require that it is matched to exactly $t_j$ points of $X$. We run the algorithm $\M$ on this instance of $\lmatch$ (line~\ref{l:M}). We define $X^{(Y,\btau)}_0$ as the set of points of $X$ matched by this algorithm. Finally, we output the solution of minimum cost (line~\ref{l:pick}).

 \begin{algorithm}[H]
  \caption{Algorithm for outlier constrained clustering.}
  \label{algo:cluster}
    {\bf Input:}  $\cI := (X,F,k,m, \chck, \cost)$\;
     Execute $\C$ on the instance $\cI' := (X,F,k+m)$ to obtain a set $C$ of $k+m$ centers. \label{l:C} \;
     Sample a set $S$ of $\lceil \frac{4\beta m \log m}{\varepsilon} \rceil$ points with replacement, each using $D^z$-sampling from $X$ w.r.t. $C$. \label{l:D2} \;
      \For{ each subset $Y \subset S, |Y| \leq m$}{
          \For {each valid tuple $\btau = (t_1, \ldots, t_{k+m})$ w.r.t. $Y$ \label{l:T}} { 
              Construct the instance $\cI^{(Y,{\bf \tau})}$  \label{l:match} \;
               Run $\M$ on $\cI^{(Y,{\bf \tau})}$ and let $X_0^{(Y, \btau)}$ be the set of matched points in $X$. \label{l:M} \;
              % Let $Z^{(F, \btau)}$ denote $F \cup X^{(F, \btau)}$. \;
               Run the algorithm $\A$ on the instance $(X \setminus (X_0^{(Y, \btau)} \cup Y), F, k, \chck, \cost)$. \label{l:rest} \;
               Let $(X^{(Y,\btau)}_1, \ldots, X^{(Y,\btau)}_k)$ be the clustering produced by $\A$. \label{l:outl} \;
          }
    }
     Let $(Y^\star, \btau^\star)$ be the pair for which $\cost(X^{(Y,\btau)}_1, \ldots, X^{(Y,\btau)}_k)$ is minimized. \label{l:pick} \;
      {\bf Output} $(X^{(Y^\star,\btau^\star)}_0, X^{(Y^\star,\btau^\star)}_1, \ldots, X^{(Y^\star,\btau^\star)}_k)$. \label{l:out}
\end{algorithm}

\section{Analysis}\label{sec:analysis}
We now analyze~\Cref{algo:cluster}. We refer to the notation used in this algorithm. Let $\cI=(X,F,k,m,\chck, \cost)$ be the instance of the outlier constrained clustering problem. Let $\opt(\cI)$ denote the optimal cost of a solution for the instance $\cI$.  Assume that the algorithm $\C$ for the unconstrained clustering problem (used in line~\ref{l:C}) is a $\beta$-approximation algorithm. We overload notation and use $\cost_{\cI'}(C)$ to denote the cost of the solution $C$ for the instance $\cI'$. Observe that the quantity $\cost_{\cI'}(C)$ can be computed as follows: each point in $X$ is assigned to the closest point in $C$, and then we compute the total cost (which could be the $\kmed$ or the $\kmean$ cost based on the value of the parameter $z$) of this assignment.
We first relate $\cost_{\cI'}(C)$ to $\opt(\cI)$. 

\begin{clm}
\label{cl:1}
$\cost_{\cI'}(C) \leq \beta \cdot \opt(\cI)$. 
\end{clm}

\begin{proof}
    Let $(X_0, X_1, ..., X_k)$ denote the optimal solution for $\cI$, where $X_0$ denotes the set of $m$ outlier points (without loss of generality, we can assume that the number of outlier points in the optimal solution is exactly $m$).
    Let $c_1, \ldots, c_k$ be the centers of the clusters $X_1, \ldots, X_k$ respectively. Consider the solution to $\cI'$ consisting of centers $C' :=X_0 \cup \{c_1, \ldots, c_k\}$. Clearly, $\cost_{\cI'}(C') \leq \opt(\cI)$ (we have inequality here because the solution $X_1, \ldots, X_k$ may not be a Voronoi partition with respect to $c_1, \ldots, c_k$). Since $\C$ is a $\beta$-approximation algorithm, we know that $\cost_{\cI'}(C) \leq \beta \cdot \cost_{\cI'}(C')$. Combining these two facts implies the desired result. 
\end{proof}

We now consider an optimal solution for the instance $\cI$: let $\Xopt_0, \Xopt_1, \ldots, \Xopt_k$ be the partition of the input points $X$ in this solution, with $\Xopt_0$ being the set of $m$  outliers. 
Depending on the distance from $C$, we divide the set $\Xopt_0$ into two subsets -- $\Xopt_F$ (``far'' points) and $\Xopt_N$ (``near'' points) as follows:

$$ \Xopt_F := \left\{ x \in \Xopt_0 | D^z(x,C) \geq \frac{\varepsilon \, \cost_{\cI'}(C)}{2 \beta m} \right\}, \quad \Xopt_N := X \setminus \Xopt_F. $$

Recall that we 
sample a set $S$ of $\frac{4 \beta m \log{m}}{\veps}$ clients using $D^z$-sampling with respect to center set $C$ (line~\ref{l:D2} in~\Cref{algo:cluster}). Note that the probability of sampling a point $x$ is given by 
\begin{align}
    \label{eq:probD}
    \frac{D^z(x,C)}{\sum_{x' \in X} D^z(x,C)} = \frac{D^z(x,C)}{\cost_{\cI'}(C)}. 
\end{align}
We first show that $S$ contains all the points in $\Xopt_F$ with high probability. 

\begin{clm}
\label{cl:F}
$\pr[\Xopt_F \subseteq S] \geq 1 - 1/m$.
\end{clm}
\begin{proof}
Inequality~\eqref{eq:probD} shows that the probability of sampling a point $x \in \Xopt_F$ is $\frac{D^z(x, C)}{\cost_{\cI'}(C)} \geq \frac{\veps}{2\beta m}$. Hence the probability that the point $x$ is not present in $S$ is at most $\left(1-\frac{\veps}{2\beta m} \right)^{\frac{4\beta m \log{m}}{\veps}} \leq \frac{1}{m^2}$. Using union bound, the probability that there is a point in $\Xopt_F$ that is not included in $S$ is at most $$\frac{|S|}{m^2} = \frac{1}{m}.$$
This shows the desired result. 
\end{proof}

\noindent
For the rest of the analysis, we condition on the event in \Cref{cl:F}, i.e., we assume $\Xopt_F \subseteq S$. We now note that the total cost of assigning $\Xopt_N$ to $C$ is  $O(\veps) \cdot \opt(\cI)$.
\begin{clm}\label{cl:3}
$\sum_{x \in \Xopt_N}  D^z(x, C) \leq \frac{\veps}{2} \cdot \opt(\cI)$.
\end{clm}
\begin{proof}
The claim follows from the following sequence of inequalities:
\begin{eqnarray*}
\sum_{x \in \Xopt_N}  D^z(x, C) < \sum_{x \in \Xopt_N} \frac{\veps \, \cost_{\cI'}(C)}{2 \beta m} \leq \sum_{x \in \Xopt_N} \frac{\veps \cdot \opt(\cI)}{2m} \leq \frac{\varepsilon}{2} \cdot \opt(\cI),
\end{eqnarray*}
where the first inequality follows from the definition of $\Xopt_N$ and the second inequality follows from~\Cref{cl:1}. 
\end{proof}

For every point in $\Xopt_N$, we identify the closest center in $C = \{c_1, \ldots, c_{m+k} \}$ (breaking ties arbitrarily). For each $j \in [k+m]$, let $\Xopt_{N,j}$  be the set of points in $\Xopt_N$ which are closest to $c_j$.  
Let $\hatt_j$ denote $|\Xopt_{N,j}|$. Consider an iteration of line~\ref{l:M}--\ref{l:outl} where $Y=\Xopt_F, \btau=(\hatt_1, \ldots, \hatt_{k+m})$. Observe that $\btau$ is valid with respect to $Y$ because $\sum_{j \in [m+k]} |\hatt_j| + |Y| = m $. Let $\hatX_1, \ldots, \hatX_{m+k}$ be the set of points assigned to $c_1, \ldots, c_{m+k}$ respectively by the algorithm $\M$. Intuitively, we would like to construct a solution where the set of outliers is given by $\hatX := \Xopt_F \cup \hatX_1 \cup \cdots \cup \hatX_{m+k}$. We now show that the set $\hatX$ is ``close'' to $\Xopt_0$, the set of outliers in the optimal solution.  In order to do this, we set up a bijection $\mu: \Xopt_0 \rightarrow \hatX$, where $\mu$ restricted to $\Xopt_F$ is identity, and $\mu$ restricted to any of the sets $\Xopt_{N,j}$ is a bijection from $\Xopt_{N,j}$ to $\hatX_{j}$. Such a function $\mu$ is possible because for each $j \in [m+k]$, $|\Xopt_{N,j}| = |\hatX_j| = \hatt_j$. We now prove this closeness property.

\begin{lemma}\label{lem:mu}
$\sum_{x \in \Xopt_0} D^z(x, \mu(x)) \leq \varepsilon \cdot z \cdot \opt(\cI)$.
\end{lemma}

\begin{figure}[t!]
\centering
\includegraphics[scale=0.30]{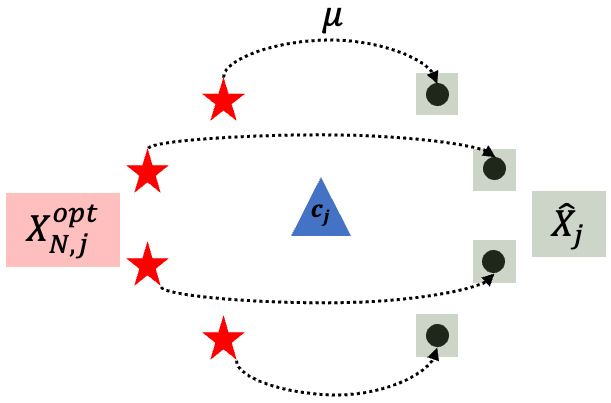}
\caption{The optimal outliers with closest center as $c_j$ ({\em see red stars}) are denoted by $X_{N, j}^{opt}$. Since we cannot distinguish them from other points near $c_j$, we find their proxies $\hat{X}_j$ ({\em see points shaded green}). Even though we show these sets as disjoint in the diagram, they may contain common points. We will designate $\hat{X}_j$ as the outlier points. This replacement of optimal outliers with their proxies may cause a loss. However, this loss can be bounded by the sum of distances between an optimal outlier and its image as per a one-to-one mapping $\mu$ ({\em see dotted arrows}) between $X_{N, j}^{opt}$ and $\hat{X}_j$.}
\label{fig:lemma-1}
\end{figure}

\begin{proof}
    We first note a useful property of the solution given by the algorithm $\M$. One of the possible solutions for the instance $\cI^{(Y, \btau)}$ could have been assigning $\Xopt_{N,j}$ to the center $c_j$. Since $\M$ is an optimal algorithm for $\lmatch$, we get 
    \begin{align}
        \label{eq:A}
        \sum_{j \in [k+m]} \sum_{x \in \hatX_j} D^z(x,c_j) \leq 
        \sum_{j \in [k+m]} \sum_{x \in \Xopt_{N,j}} D^z(x,c_j) = \sum_{x \in \Xopt_N} D^z(x,C) \leq \frac{\varepsilon}{2} \cdot \opt(\cI), 
    \end{align}
    where the last inequality follows from~\Cref{cl:3}. Now, 
    \begin{align}
        \notag
    \sum_{x \in \Xopt_0} D^z(x,\mu(x)) & =  \sum_{x \in \Xopt_N} D^z(x,\mu(x)) = \sum_{j \in [k+m]} \sum_{x \in \Xopt_{N,j}} D^z(x,\mu(x)) \\
    & \stackrel{\eqref{eq:tr}}{\leq} z \cdot \sum_{j \in [k+m]} \sum_{x \in \Xopt_{N,j}} \left( D^z(x,c_j) + D^z(c_j,\mu(x)) \right),
    \end{align}
    where the first equality follows from the fact that $\mu$ is identity on $\Xopt_F$. Since $\mu$ is a bijection from $\Xopt_{N,j}$ to $\hatX_j$, the above can also be written as 
    $$ z \cdot  \sum_{j \in [k+m]} \sum_{x \in \Xopt_{N,j}} D^z(x,c_j) + 
    z \cdot  \sum_{j \in [k+m]} \sum_{x \in \hatX_{j}} D^z(x,c_j) \leq z \cdot \varepsilon \; \opt(\cI),$$
    where the last inequality follows from~\Cref{cl:3} and~\eqref{eq:A}. This proves the desired result. 
\end{proof}

The mapping $\mu$ described above may have the following undesirable property: there could be a point $x \in \Xopt_0 \cap \hatX$ such that $\mu(x) \neq x$. This could happen if $x \in \Xopt_{N,j}$ and $x \in \hatX_i$ where $i \neq j$. We now show that $\mu$ can be modified to another bijection $\hmu: \Xopt_0 \rightarrow \hatX$ which is identity on $\Xopt_0 \cap \hatX.$ Note that the mapping $\hmu$ is only needed for the analysis of the algorithm. 
\begin{lemma}
    \label{lem:hmu}
    There is a bijection $\hmu: \Xopt_0 \rightarrow \hatX$ such that $\hmu(x) = x$ for all $x \in \Xopt_0 \cap \hatX$ and 
    $$ \sum_{x \in \Xopt_0} D^z(x,\hmu(x)) \leq m^{z-1} \, \varepsilon \cdot z \cdot \opt(\cI). $$
\end{lemma}

\begin{figure}[t!]
\centering
\includegraphics[scale=0.30]{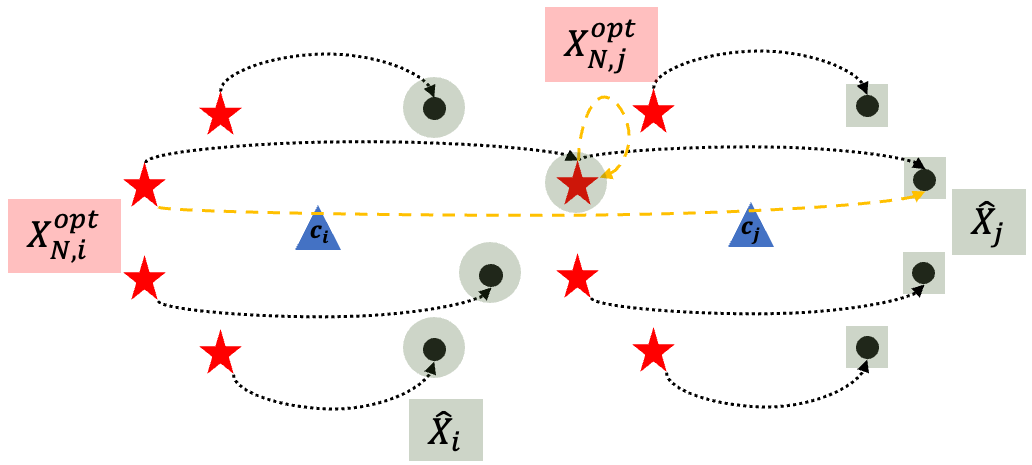}
\caption{We want to designate the proxies as outliers instead of their pre-images ({\em as per the mapping $\mu$ defined in Figure~\ref{fig:lemma-1}}). The penalty of this replacement will not be too much, as per Lemma~\ref{lem:mu}. However, there is an issue with this plan if a proxy point in $\hat{X}_i$ is also an optimal outlier in $X^{opt}_{N, j}$ for $i \neq j$ ({\em see the star shaded with a green circle}). In this case, we modify the one-to-one mapping $\mu$ to $\hat{\mu}$ by tracing the mapping $\mu$ starting from an optimal outlier to a non-outlier ({\em see star on the left to point on the right}). We map the extreme points to each other and map the intermediate points to themselves ({\em see yellow dashed lines}). The penalty of this mapping will now depend on the distance between the extreme points, but that can be bounded by applying the approximate triangle inequality along the path.}
\label{fig:lemma-2}
\end{figure}

\begin{proof}
    We construct a directed graph $H=(V_1, E_1)$ where $V_1 = \Xopt_0 \cup \hatX$. For every $x \in \Xopt_0$, we add the directed arc $(x, \mu(x))$ to $E_1$. Observe that a self loop in $H$ implies that $\mu(x)=x$. Every vertex in $\Xopt_0 \setminus \hatX$ has 0 in-degree and out-degree 1; whereas a vertex in $\hatX \setminus \Xopt_0$ has in-degree 1 and 0 out-degree. Vertices in $\hatX \cap \Xopt_0$ have exactly one incoming and outgoing arc (in case of a self-loop, it counts towards both the in-degree and the out-degree of the corresponding vertex). 

     The desired bijection $\hmu$ is initialized to $\mu$. Let $\cost(\hmu)$ denote $\sum_{x \in \Xopt_0} D^z(x, \hmu(x))$; define $\cost(\mu)$ similarly. 
    It is easy to check $H$ is vertex disjoint union of directed cycles and paths. In case of a directed cycle $C$ on more than 1 vertex, it must be the case that each of the vertices in $C$ belong to $\hatX \cap \Xopt_0$. In this case, we update $\hmu$ be defining $\hmu(x)=x$ for each $x \in C$. Clearly
    this can only decrease $\cost(\hmu)$.
    Let $P_1, \ldots, P_l$ be the set of directed paths in $H$. For each path $P_j$, we perform the following update: let $P_j$ be a path from $a_j$ to $b_j$. We know that $a_j \in \Xopt \setminus \hatX$, $b_j \in \hatX \setminus \Xopt_0$ and each internal vertex of $P_j$ lies in $\hatX \cap \Xopt_0$. We update $\hmu$ as follows; $\hmu(a_j) = b_j$ and $\hmu(v)=v$ for each internal vertex $v$ of $P_j$. The overall increase in $\cost(\hmu)$ is equal to 
    \begin{align}
        \label{eq:inc}
    \sum_{j \in [l]} \left( D^z(a_j,b_j) - \sum_{i=1}^{n_j} D^z(v_j^i, v_j^{i-1}) \right),
    \end{align}
    where $a_j = v_j^0,  v_j^1, \ldots, v_j^{n_j} = b_j$ denotes the sequence of vertices in $P_j$. If $z =1$, triangle inequality shows that the above quantity is at most 0. In case $z=2$, 
    $$D^2(a_j,b_j) \leq n_j \left( \sum_{i=1}^{n_j} D^2(v_j^i, v_j^{i-1}) \right),$$
    and so the quantity in~\eqref{eq:inc} is at most $(n_j-1) \sum_{i=1}^{n_j} D^2(v_j^i, v_j^{i-1}).$

    It follows that $\cost(\hmu) \leq m^{z-1}   \cost(\mu).$ The desired result now follows from~\Cref{lem:mu}. 
\end{proof}

We run the algorithm $\A$ on the outlier-free constrained clustering instance $\cI'' = (X \setminus \hatX, F, k,\chck, \cost)$ (line~\ref{l:rest} in~\Cref{algo:cluster}). Let $\opt(\cI'')$ be the optimal cost of a solution for this instance. The following key lemma shows that $\opt(\cI'')$ is close to $\opt(\cI)$. 

\begin{lemma}
    \label{lem:close}
    $\opt(\cI'') \leq (1+\varepsilon^{\frac{1}{z}} (4m+1)^{z-1}) \opt(\cI).$
\end{lemma}
\begin{proof}
    We shall use the solution $(\Xopt_0, \ldots, \Xopt_k)$ to construct a feasible solution for $\cI''$. 
    For each $j \in [k]$, let $Z_j$ denote $\Xopt_j \cap \hatX$. Let $\hmu^{-1}(Z_j)$ denote the pre-image under $\hmu$ of $Z_j$. Since $Z_j \subseteq \hatX \setminus \Xopt_0$, $\hmu^{-1}(Z_j) \subseteq \Xopt_0 \setminus \hatX$.  For each $j \in [k]$, define $ X'_j := (\Xopt_j \setminus Z_j) \cup \hmu^{-1}(Z_j).$
   \begin{clm}
   \label{cl:Z}
   $ \bigcup_{j=1}^k X_j' = X \setminus \hatX.$
   \end{clm}
   \begin{proof}
      For any $j \in [k]$, we have already argued that $\hmu^{-1}(Z_j) \subseteq \Xopt_0 \setminus\hatX \subseteq X \setminus \hatX$. Clearly, $\Xopt_j \setminus Z_j \subseteq X \setminus \hatX$. Therefore $X_j' \subseteq X \setminus \hatX$. Therefore, $\cup_{j \in [k]} X_j' \subseteq X \setminus \hatX$.
      Since $|X_j'| = |\Xopt_j|$, 
      $$ \sum_{j \in [k]} |X_j'| = n-m = |X \setminus \hatX|.$$ This proves the claim.\end{proof}
      The above claim implies that $(X_1', \ldots, X_k')$ is a partition of $X \setminus \hatX$.  Since $|X_j'| = |\Xopt_j|$ for all $j \in [k]$ and the function $\chck$ only depends on the cardinality of the sets in the partition, $(X_1', \ldots, X_k')$ is a feasible partition (under $\chck$) of $X \setminus \hatX$.   In the optimal solution for $\cI$, let $\fopt_1, \ldots, \fopt_k$ be the $k$ centers corresponding to the clusters $\Xopt_1, \ldots, \Xopt_k$ respectively. Now, 
      \begin{align}
          \label{eq:step1}
          \opt(\cI'') & \leq \cost(X_1', \ldots, X_k') 
           \leq \sum_{j \in [k]} \sum_{x \in X_j'} D^z(x,\fopt_j) 
      \end{align}
      For each $j \in [k]$, we estimate the quantity $\sum_{x \in X_j'} D^z(x,\fopt_j)$. Using the definition of $X_j'$ and triangle inequality, this quantity can be expressed as 
      \begin{align}
          \sum_{x \in \Xopt_j \setminus Z_j} D^z(x,\fopt_j) + \sum_{x \in \hmu^{-1}(Z_j)} D^z(x, \fopt_j) & \leq 
          \sum_{x \in \Xopt_j \setminus Z_j} D^z(x,\fopt_j) + \sum_{x \in \hmu^{-1}(Z_j)} \left(D(x, \hmu(x)) + D(\hmu(x), \fopt_j) \right)^z
          \label{eq:step2}
      \end{align}
      When  $z=1$, the above is at most (replacing $x$ by $\hmu(x)$ in the second expression on RHS)
      $$ \sum_{x \in \Xopt_j} D(x,\fopt_j) + \sum_{x \in Z_j} D(x, \hmu(x)). $$
      Using this bound in~\eqref{eq:step1}, we see that 
      $$ \opt(\cI'') \leq \opt(\cI) + \sum_{x \in \Xopt_0} D(x,\hmu(x)) \leq (1+\varepsilon) \opt(\cI),$$
      where the last inequality follows from~\Cref{lem:hmu}. This proves the desired result for $z=1$. 
      When $z=2$, we use the fact that for any two reals $a,b$, 
      $$(a+b)^2 \leq (1+\sqrt{\varepsilon}) a^2 + b^2 \left( 1 + \frac{1}{\sqrt{\varepsilon}} \right).$$
      Using this fact, the expression in the RHS of~\eqref{eq:step2} can be upper bounded by 
      $$(1+\sqrt{\varepsilon}) \sum_{x \in \Xopt_j} D^2(x,\fopt_j) + \left( 1 + \frac{1}{\sqrt{\varepsilon}} \right) \sum_{x \in Z_j} D^2(x, \hmu(x)).$$
      Substituting this expression in~\eqref{eq:step1} and using~\Cref{lem:hmu}, we see that 
      $$ \opt(\cI'') \leq (1+\sqrt{\varepsilon}) \opt(\cI) + 4m \sqrt{\varepsilon} \opt(\cI).$$
      This proves the desired result for $z=2$.  
\end{proof}

The approximation preserving properties of Theorem~\ref{thm:main} follow from the above analysis. For the \kmean problem, since the approximation term is $(1+\sqrt{\veps}(4m+1))$, we can replace $\veps$ with $\veps^2/(4m+1)^2$ in the algorithm and analysis to obtain a $(1+\veps)$ factor. Let us quickly check the running time of the algorithm. The algorithm first runs $\mathcal{C}$ that takes $T_{\mathcal{C}}(n)$ time. 
This is followed by $D^z$-sampling $O(\frac{m^{z+1}\log{m}}{\veps^z})$ points, which takes $O(n \cdot (k + \frac{m^{z+1}\log{m}}{\veps^z}))$ time. 
The number of iterations of the for-loops is determined by the number of subsets of $S$, which is $\sum_{i=0}^{m} \binom{|S|}{i} = \left( \frac{m}{\veps}\right)^{O(m)}$, and the number of possibilities for $\tau$, which is at most $\binom{2m+k-1}{m} = (m+k)^{O(m)}$. This gives the number of iterations $q = f(k, m, \veps) = \left(\frac{k+m}{\veps} \right)^{O(m)}$. 
In every iteration, in addition to running $\mathcal{A}$, we solve a weighted b-matching problem on a bipartite graph $(L\cup R, E)$ where $R$ has $(k+m)$ vertices (corresponding to the $k+m$ centers in the center set $C$) and $L$ has at most $(k+m)\cdot m$ vertices (considering $m$ closest clients for every center is sufficient which can be found using a pre-processing step). So, every iteration costs $T_{\mathcal{A}}(n) + O((k+m)^3m^2)$ time. This gives the running time expression in Theorem~\ref{thm:main}.

\subsection*{Extension to labelled version}\label{sec:labelled}
In this section, we extend~\Cref{algo:cluster} to the setting where points in $X$ have labels. Recall that the label assignment is specified by a mapping $\sigma: X \rightarrow L,$ where $L$ is the set of labels. Given a partitioning $X_0, X_1, \ldots, X_k$ and locations $f_1, \ldots, f_k$,
% from a finite set $L$ and the 
$\chck(X_0, \ldots, X_k, f_1, \ldots, f_k, \sigma)$  depends on $|\sigma^{-1}(l) \cap X_j|$ for each $l \in L, j \in [k]$.  Analogous to the assumptions for the unlabelled version, we assume the existence of the algorithms $\cal C$, $\cal A$ and $\cal M$. Note that the algorithm $\cal A$ for the outlier-free version takes as input a tuple $(X',F',k,\sigma, \chck, \cost)$.

The overall structure of~\Cref{algo:cluster} remains unchanged.  However, after guessing the subset $Y$, we need to guess not only the number of outlier points close to each of the centers in $C$, but also the labels of such points. 
%; we just indicate the changes needed in this algorithm.  
Motivated by this, we give the following key definition. 
Given a non-negative integer $p$, a {\em label partition} of $p$ is defined as a tuple $\psi = (q_1, \ldots, q_{|L|})$ such that $\sum_i q_i = p$. The intuition is that given a set $S$ of size $p$, $q_1$ points get the first label in $L$, $q_2$ points in $S$ get the second label in $L$, and so on. 

The algorithm is described in~\Cref{algo:clusterlabel}. As before, we execute $\cal C$ on the underlying unconstrained instance $\cI'$ to obtain a set $C$ of $k+m$ centers (line~\ref{l:Cl}), and then sample the subset $S$ using $D^z$-sampling (line~\ref{l:D2l}). 
 Now, given a subset $Y$ of $S$, define {\em a valid tuple} $\btau$ w.r.t. $Y$ as a tuple 
$((t_1, \psi_1), \ldots, (t_{k+m}, \psi_{k+m})),$ where 
(i) $\sum_j t_j + |Y| =m$, and (ii) $\psi_j$ is a label partition of $t_j$. As in line~\ref{l:T} in~\Cref{algo:cluster}, we iterate over all such valid tuples in line~\ref{l:Tl}. The definition
of the instance $\cI^{(Y, \tau)}$ changes as follows: as before, we construct a bipartite graph with the set of $(k+m)$ centers $C$ on one side and $X$ on the other side. The weight of an edge between a center $v \in C$ and point $w \in X$ is set to $D^z(v,w)$.  
%The definition of a solution to the \lmatch instance $\cI^{(Y, \tau)}$ changes as follows. 
Let $\psi_j:=(n_j^1, \ldots, n_j^{\ell})$, where $\ell = |L|$. 
But now,  a solution to the instance $\cI^{(Y, \tau)}$ needs to satisfy the condition that for each point $c_j \in C$ and each label $l \in L$, exactly $n_j^l$ points in $X$ are matched to $c_j$. Note that this also implies that exactly $t_j$ points are matched to $c_j$. This matching problem can be easily reduced to \lmatch -- make $|L|$ copies of each vertex $c_j$. Call these copies $c_j^1, \ldots, c_j^{\ell}$. Now for each $l \in [\ell]$, we add edges from $c_j^l$ to points in $X$ whose label is $l$. Thus, the edges going out of $c_j$ in the original graph get partitioned into $\ell$ groups. Now we require that $c_j^l$ is matched to exactly $n_j^l$ vertices. 
% by making $t_j$ copies of each point $c_j$, and for each label $l$, adding edges between $n_j^l$ distinct copies of $c_j$ to vertices of label $l$ only. 
The rest of the details of~\Cref{algo:clusterlabel} are same as that of~\Cref{algo:cluster}. Note that the running time of the algorithm changes because we now have to iterate over all partitions of each of the numbers $t_j$.

\begin{algorithm}[H]
  \caption{Algorithm for outlier constrained clustering when points have labels.}
  \label{algo:clusterlabel}
    {\bf Input:}  $\cI := (X,F,k,m, \sigma, \chck, \cost)$\;
     Execute $\C$ on the instance $\cI' := (X,F,k+m)$ to obtain a set $C$ of $k+m$ centers. \label{l:Cl} \;
     Sample a set $S$ of $\lceil \frac{4\beta m \log m}{\varepsilon} \rceil$ points with replacement, each using $D^z$-sampling from $X$ w.r.t. $C$. \label{l:D2l} \;
      \For{ each subset $Y \subset S, |Y| \leq m$}{
          \For {each valid tuple $\btau = ((t_1,\psi_1), \ldots, (t_{k+m},\psi_{k+m}))$ w.r.t. $Y$ \label{l:Tl}} { 
              Construct the instance $\cI^{(Y,{\bf \tau})}$  \label{l:matchl} \;
               Run $\M$ on $\cI^{(Y,{\bf \tau})}$ and let $X_0^{(Y, \btau)}$ be the set of matched points in $X$. \label{l:Ml} \;
              % Let $Z^{(F, \btau)}$ denote $F \cup X^{(F, \btau)}$. \;
               Run the algorithm $\A$ on the instance $(X \setminus (X_0^{(Y, \btau)} \cup Y), F, k,\sigma, \chck, \cost)$. \label{l:restl} \;
               Let $(X^{(Y,\btau)}_1, \ldots, X^{(Y,\btau)}_k)$ be the clustering produced by $\A$. \label{l:outll} \;
          }
    }
     Let $(Y^\star, \btau^\star)$ be the pair for which $\cost(X^{(Y,\btau)}_1, \ldots, X^{(Y,\btau)}_k)$ is minimized. \label{l:pickl} \;
      {\bf Output} $(X^{(Y^\star,\btau^\star)}_0, X^{(Y^\star,\btau^\star)}_1, \ldots, X^{(Y^\star,\btau^\star)}_k)$. \label{l:out-label}
\end{algorithm}

 The analysis of the algorithm proceeds in an analogous manner as that of~\Cref{algo:cluster}. We just need to consider the iteration of the algorithm, where we correctly guess the size of each of the sets $\Xopt_{N,j}$ and the number of points of each label in this set.

\section{Conclusion and Open Problems}
In this work, we give an approximation preserving reduction from the outlier version of the $k$-means/median problems to their outlier-free versions. The main idea is to find a list of possibilities $O_1, ..., O_q$ for the $m$ outlier points from the dataset $X$, solve the outlier-free version on instances $X\setminus O_1$, ..., $X\setminus O_q$ and pick the best solution. Note that there is a trivial reduction with $q = \binom{n}{m}$ ({\em try all choices of $m$ outliers from the dataset $X$, with $n= |X|$}), which is prohibitively large. To obtain a much smaller set of possibilities, which is independent of the data size $n$, we try to find suitable replacements for the $m$ optimal outlier points. The issue is that such a replacement should not increase the optimal cost of the outlier version too much. To ensure this, we find a center set $C$ with $(k+m)$ centres that approximates the cost of the outlier version. The key insight is that the optimal outliers that are far away from $C$ can be found using distance-based sampling, and for the ones that are sufficiently close to $C$ ({\em we have a good handle on the closeness since $C$ approximates the cost}), we find replacement points that increase the optimal cost by a small amount.

Moreover, our reduction works within a very general framework for modeling constrained versions of these clustering problems, which enables us to obtain approximation results for a wide range of constrained clustering problems. One future direction is to explore which other constrained clustering problems fit our framework and, hence, can benefit from our reduction. The more important question is related to the efficiency of the reduction. The key quantity in our reduction is $q \coloneqq f(k, m, \veps) = \left( \frac{k+m}{\veps}\right)^{O(m)}$, the number of instances of the outlier-free problem that must be solved to obtain a good solution for the outlier version. Whether this can be improved and to what extent it can be improved are interesting open problems.
%Improving this bound is a challenging open problem. 

%\section*{Conflict of Interest Statement}
%The authors have no competing interests to declare that are relevant to the content of this article.

%%
%% Bibliography
%%

%% Please use bibtex, 
\addcontentsline{toc}{section}{References}
\bibliography{references}

\appendix 

\section{Tables}
\begin{table}
\begin{adjustbox}{width=\columnwidth,center}
\centering
\setcellgapes{1ex}\makegapedcells
\begin{tabular}{|l|l|}
\hline
{\bf Problem} & {\bf Description} \\ \hline
\makecell[l]{Unconstrained $k$-median \\ ({\it Constraint type: unconstrained})} & \makecell[l]{{\it Input}: $(F, X, k)$ \\ {\it Output}: $(X_1, ..., X_k, f_1, ..., f_k)$ \\ {\it Constraints}: None, i.e., $\chck(X_1, ..., X_k, f_1, ..., f_k)$ always equals 1.\\ {\it Objective}: Minimise $\sum_{i} \sum_{x \in X_i} D(x, f_i)$. \\ (This includes various versions corresponding to specific metrics such as \\
\ \ Ulam metric on permutations, metric spaces with constant doubling dimension etc.)} \\ \hline
\makecell[l]{Fault-tolerant $k$-median  \\ ({\it Constraint type: unconstrained}\\\ {\it but labelled}) \\\\ \cite{fault:kmedian_2014_non_uniform_haji_li_SODA,fault:outlier_kmedian_2020_Varadarajan}} & \makecell[l]{{\it Input}: $(F, X, k)$ and a number $h(x) \leq k$ for every facility $x \in X$  \\ {\it Output}: $(f_1, ..., f_k)$ \\ {\it Constraints}: None.\\ {\it Objective}: Minimise $ \sum_{x \in X} \sum_{j = 1}^{h(x)} D(x, f_{\pi_x(j)})$, \\
\hspace*{0.6in} where $\pi_x(j)$ is the index of $j^{th}$ nearest center to $x$ in $(f_1, ..., f_k)$\\ ({\it Label}: $h(x)$ may be regarded as the label of the client $x$. So, the number of distinct labels $\ell \leq k$.)} \\ \hline
\makecell[l]{Balanced $k$-median \\ ({\it Constraint type: size}) \\\\ \cite{rgather:k_center_2010_Aggarwal,ding20}} & \makecell[l]{{\it Input}: $(F, X, k)$ and integers $(r_1, ..., r_k)$, $(l_1, ..., l_k)$,  \\ {\it Output}: $(X_1, ..., X_k, f_1, ..., f_k)$ \\ {\it Constraints}: $X_i$ should have at least $r_i$ and at most $l_i$ clients, \\
\hspace*{0.75in} i.e., $\chck(X_1, ..., X_k, f_1, ..., f_k) = 1$ iff $\forall i, r_i \leq |X_i| \leq l_i$ .\\ {\it Objective}: Minimise $\sum_{i} \sum_{x \in X_i} D(x, f_i)$. \\
(Versions corresponding to specific values of $r_i$'s and $l_i$'s are known by different names. \\
\ \ The version corresponding to $l_1 = ... = l_k = |X|$ is called the $r$-gather problem and \\
\ \ the version where $r_1= ... = r_k = 0$ is called the $l$-capacity problem.)} \\ \hline
\makecell[l]{Capacitated $k$-median \\ ({\it Constraint type: center + size}) \\\\ \cite{cl19}} & \makecell[l]{{\it Input}: $(F, X, k)$ and with capacity $s(f)$ for every facility $f \in F$  \\ {\it Output}: $(X_1, ..., X_k, f_1, ..., f_k)$ \\ {\it Constraints}: The number of clients, $X_i$, assigned to $f_i$ is at most $s(f_i)$, \\
\hspace*{0.75in} i.e., $\chck(X_1, ..., X_k, f_1, ..., f_k) = 1$  iff $\forall i, |X_i| \leq s(f_i)$ .\\ {\it Objective}: Minimise $\sum_{i} \sum_{x \in X_i} D(x, f_i)$.} \\ \hline
\makecell[l]{Matroid $k$-median \\ ({\it Constraint type: center}) \\\\ \cite{kknss11,vincent_hardness_2019}} & \makecell[l]{{\it Input}: $(F, X, k)$ and a Matroid on $F$  \\ {\it Output}: $(X_1, ..., X_k, f_1, ..., f_k)$ \\ {\it Constraints}: $\{f_1, ..., f_k\}$ is an independent set of the Matroid, \\
%The number of clients, $X_i$, assigned to $f_i$ is at most $s(f_i)$, \\
\hspace*{0.75in} i.e., $\chck(X_1, ..., X_k, f_1, ..., f_k) = 1$  iff $\{f_1, ..., f_k\}$ is an independent set of the Matroid .\\ {\it Objective}: Minimise $\sum_{i} \sum_{x \in X_i} D(x, f_i)$.} \\ \hline
\makecell[l]{Strongly private $k$-median \\ ({\it Constraint type: label + size}) \\\\ \cite{rosner18}} & \makecell[l]{{\it Input}: $(F, X, k)$ and numbers $(l_1, ..., l_w)$. Each client has a label $\in \{1, ..., w\}$.   \\ {\it Output}: $(X_1, ..., X_k, f_1, ..., f_k)$ \\ {\it Constraints}: Every $X_i$ has at least $l_j$ clients with label $j$, \\ 
\hspace*{0.75in} i.e., $\chck(X_1, ..., X_k, f_1, ..., f_k) = 1$  iff $\forall i, j, |X_i \cap S_j| \geq l_j$, \\
\hspace*{0.75in} where $S_j$ is the set of clients with label $j$ .\\ {\it Objective}: Minimise $\sum_{i} \sum_{x \in X_i} D(x, f_i)$. \\ ({\it Labels}: The number of distinct labels $\ell = w$).} \\ \hline
\makecell[l]{$l$-diversity $k$-median \\ ({\it Constraint type: label + size}) \\\\ \cite{bercea19}} & \makecell[l]{{\it Input}: $(F, X, k)$ and a number $l>1$. Each client has one colour from $\in \{1, ..., w\}$   \\ {\it Output}: $(X_1, ..., X_k, f_1, ..., f_k)$ \\ {\it Constraints}: The fraction of clients with colour $j$ in every $X_i$ is at least $1/l$, \\ 
\hspace*{0.75in} i.e., $\chck(X_1, ..., X_k, f_1, ..., f_k) = 1$  iff $\forall i, j, |X_i \cap S_j| \leq |X_i|/l$, \\
\hspace*{0.75in} where $S_j$ is the set of clients with colour $j$ .\\ {\it Objective}: Minimise $\sum_{i} \sum_{x \in X_i} D(x, f_i)$. \\ ({\it Labels}: Each colour can be regarded as a label and hence the number of distinct labels $\ell = w$).} \\ \hline
\makecell[l]{Fair $k$-median \\ ({\it Constraint type: label + size}) \\\\ \cite{bercea19,bera19}} & \makecell[l]{{\it Input}: $(F, X, k)$ and fairness values $(\alpha_1, ..., \alpha_w), (\beta_1, ..., \beta_w)$. Each client has colours from $\in \{1, ..., w\}$   \\ {\it Output}: $(X_1, ..., X_k, f_1, ..., f_k)$ \\ {\it Constraints}: The fraction of clients with colour $j$ in every $X_i$ is between $\alpha_j$ and $\beta_j$, \\
\hspace*{0.75in} i.e., $\chck(X_1, ..., X_k, f_1, ..., f_k) = 1$ \\ \hspace*{0.8in} iff $\forall i, j, \alpha_j |X_i| \leq |X_i \cap S_j| \leq \beta |X_i|$, where $S_j$ is the set of clients with colour $j$ .\\ {\it Objective}: Minimise $\sum_{i} \sum_{x \in X_i} D(x, f_i)$. \\ (There are two versions: (i) each client has a unique label, and (ii) a client can have multiple labels.)\\ ({\it Labels}: For the first version $\ell=w$ and for the second version $\ell = 2^w$.)} \\ \hline
\end{tabular}
\end{adjustbox}
\caption{The table defines various outlier-free versions of the constrained $\kmed$ problem. The $\kmean$ versions are defined similarly using $D^2$ instead of $D$. We include a few references. The problems are categorized based on the {\it type} of constraints. There are three main types of constraints (i) {\it size} (constraints on the cluster size), (ii) {\it center} (constraints on the points a center can service), and (iii) {\it label} (constraints on the label of points in clusters). A constrained problem can have a combination of these constraint types.}\label{table:1}
\end{table}

\begin{table}
\begin{adjustbox}{width=\columnwidth,center}
\centering
\setcellgapes{1ex}\makegapedcells
\begin{tabular}{|l|l|l|l|l|}
\hline
\multirow{2}{*}{\bf Problem} & \multirow{2}{*}{\bf Outlier-free} & \multicolumn{3}{l|}{\bf ~~~~~~~ Outlier version}\\ \cline{3-5}
& & \makecell[c]{\cite{gjk20}} & \makecell[c]{\cite{aisx23}} & {\bf This work} \\ \hline
\makecell[l]{Euclidean $k$-means (i.e., $F=\mathbb{R}^d, X \subset \mathbb{R}^d$)} & \makecell[c]{$(1+\veps)$ \\ \cite{bjk18}} &  \makecell[c]{$\times$} & \makecell[c]{$(1+\veps)$} & \makecell[c]{$(1+\veps)$} \\ \hline
\makecell[l]{$k$-median} & \makecell[c]{$\left(1+\frac{2}{e} + \veps \right)$ \\ \cite{vincent_hardness_2019}} &  \makecell[c]{$(3+\veps)$} & \makecell[c]{$\left(1 + \frac{2}{e} + \veps\right)$} & \makecell[c]{$\left(1 + \frac{2}{e} + \veps\right)$} \\ \hline
\makecell[l]{$k$-means} & \makecell[c]{$\left(1+\frac{8}{e} + \veps \right)$ \\ \cite{vincent_hardness_2019}} &  \makecell[c]{$(9+\veps)$} & \makecell[c]{$\left(1 + \frac{8}{e} + \veps\right)$} & \makecell[c]{$\left(1 + \frac{8}{e} + \veps\right)$} \\ \hline
\makecell[l]{$k$-median/means in metrics: \\ (i) constant doubling dimension\\(ii) metrics induced by graphs of bounded treewidth \\ (iii) metrics induced by graphs that exclude a fixed \\ graph as a minor} & \makecell[c]{$\left(1+ \veps \right)$ \\ \cite{css21}} &  \makecell[c]{\colorbox{black!10}{$(3+\veps)$} \\ $k$-median \\ \colorbox{black!10}{$(9+\veps)$} \\ $k$-means} & \makecell[c]{$\left(1 + \veps\right)$} & \makecell[c]{$\left(1 + \veps\right)$} \\ \hline
\makecell[l]{Matroid $k$-median} & \makecell[c]{$\left(2+ \veps \right)$ \\ \cite{vincent_hardness_2019}} &  \makecell[c]{\colorbox{black!10}{$(3+\veps)$}} & \makecell[c]{$\left(2 + \veps\right)$} & \makecell[c]{$\left(2 + \veps\right)$} \\ \hline
\makecell[l]{Colourful $k$-median} & \makecell[c]{$\left(1+ \frac{2}{e} + \veps \right)$ \\ \cite{vincent_hardness_2019}} &  \makecell[c]{\colorbox{black!10}{$(3+\veps)$}} & \makecell[c]{$\left(1 + \frac{2}{e} + \veps\right)$} & \makecell[c]{$\left(1 + \frac{2}{e} +  \veps\right)$} \\ \hline
\makecell[l]{Ulam $k$-median (here $F=X$)} & \makecell[c]{$\left(2 - \delta \right)$ \\ \cite{cdk23}} & \makecell[c]{\colorbox{black!10}{$(2+\veps)$}} & \makecell[c]{$\times$} & \makecell[c]{\colorbox{RoyalBlue!30}{$\left(2-\delta\right)$}} \\ \hline
\makecell[l]{Euclidean Capacitated $k$-median/means} & \makecell[c]{$\left(1 + \veps \right)$ \\ \cite{cl19}} & \makecell[c]{$\times$} & \makecell[c]{$\times$} & \makecell[c]{\colorbox{RoyalBlue!30}{$\left(1+\veps\right)$}} \\ \hline
\makecell[l]{Capacitated $k$-median\\Capacitated $k$-means} & \makecell[c]{$\left(3 + \veps \right)$\\ $(9+\veps)$ \\ \cite{cl19}} & \makecell[c]{$\times$\\$\times$} & \makecell[c]{$\times$\\$\times$} & \makecell[c]{\colorbox{RoyalBlue!30}{$\left(3+\veps\right)$}\\\colorbox{RoyalBlue!30}{$(9+\veps)$}} \\ \hline
\makecell[l]{Uniform/non-uniform $r$-gather $k$-median/means\\(uniform implies $r_1=r_2=...=r_k$)} & \makecell[c]{} &  \makecell[c]{} & \makecell[c]{} & \makecell[c]{} \\ 
\makecell[l]{Uniform/non-uniform $l$-capacity $k$-median/means\\(uniform implies $l_1=l_2=...=l_k$)} & \makecell[c]{} & \makecell[c]{} & \makecell[c]{} & \makecell[c]{} \\ 
\makecell[l]{Uniform/non-uniform balanced $k$-median/means\\(uniform implies $r_1=r_2=...=r_k$ and $l_1=l_2=...=l_k$)} & \makecell[c]{$(3+\veps)$ \\ ($k$-median)} & \makecell[c]{$(3+\veps)$\\ ($k$-median)} & \makecell[c]{$\times$} & \makecell[c]{$(3+\veps)$ \\ ($k$-median)} \\ 
\makecell[l]{Uniform/non-uniform fault tolerant $k$-median/means\\ (uniform implies same $h(x)$ for every $x$)} & \makecell[c]{$(9+\veps)$ \\ ($k$-means)} & \makecell[c]{$(9+\veps)$ \\ ($k$-means)} & \makecell[c]{$\times$} & \makecell[c]{$(9+\veps)$ \\ ($k$-means)} \\ 
\makecell[l]{Strongly private $k$-median/means} & \makecell[c]{\cite{gjk20}} &  \makecell[c]{} & \makecell[c]{} & \makecell[c]{} \\ 
\makecell[l]{$l$-diversity $k$-median/means} & \makecell[c]{} &  \makecell[c]{} & \makecell[c]{} & \makecell[c]{} \\ 
\makecell[l]{Fair $k$-median/means} &  &  \makecell[c]{} & \makecell[c]{} & \makecell[c]{} \\ \hline
\end{tabular}
\end{adjustbox}
\caption{A $\times$ means that the techniques are not known to apply to the problem. The new results that do not follow from the previously known results are shaded \colorbox{RoyalBlue!30}{~}. The results that were not explicitly reported but follow from the techniques in the paper are shaded \colorbox{black!10}{~}
%For the Ulam $k$-median problem, although the $(2+\veps)$-approximation result was not explicitly reported in \cite{gjk20}, the result should also follow from their techniques. 
The techniques of \cite{aisx23} do not apply to the Ulam $k$-median problem since the outlier-free algorithm works on unweighted instances. Note that all the FPT $(3+\veps)$ and $(9+\veps)$ approximations for the outlier-free versions (leftmost column) in the last row follow from the outlier-free results in \cite{gjk20}. However, the approximation guarantees in the rightmost column depend on those in the leftmost. This means, unlike the rigid $(3+\veps)$ and $(9+\veps)$ approximation of \cite{gjk20} in the middle column, the approximation guarantee in the rightmost column will improve with every improvement in the leftmost.}\label{table:2}
\end{table}

\end{document}